\documentclass[11pt]{article}
\usepackage[letterpaper,margin=1in]{geometry}

\usepackage{amsmath}
\usepackage{amssymb}
\usepackage{amsthm}
\usepackage{booktabs}
\usepackage{enumitem}
\usepackage{xcolor}
\usepackage{tikz}
\usepackage{eso-pic}
\usepackage{url}
\usetikzlibrary{positioning,shapes.geometric,arrows.meta,automata}

\providecommand{\citenamefont}[1]{#1}

\newtheorem{theorem}{Theorem}[section]
\newtheorem{lemma}[theorem]{Lemma}
\newtheorem{proposition}[theorem]{Proposition}

\newtheorem{definition}[theorem]{Definition}

\newcommand{\fito}{\textsc{fito}}
\newcommand{\oae}{\textsc{oae}}
\newcommand{\dco}{\textsc{dco}}
\newcommand{\ico}{\textsc{ico}}
\newcommand{\ilt}{\textsc{ilt}}
\newcommand{\pif}{\textsc{pif}}
\newcommand{\kbp}{\textsc{kbp}}

\title{The Semantic Arrow of Time, Part~II: The Semantics of Open Atomic Ethernet}
\author{Paul Borrill \\ D\AE D\AE LUS \\ \texttt{paul@daedaelus.com} \\ ORCID: 0000-0002-7493-5189}
\date{March 2026}

\begin{document}
\maketitle

\begin{abstract}
\noindent
This is the second of five papers comprising \emph{The Semantic Arrow of Time}.
Part~I~\cite{borrill2026-partI} established that computing's arrow of time is
semantic rather than thermodynamic, and that the Forward-In-Time-Only (\fito{})
assumption constitutes a category mistake.  This paper develops the constructive
alternative.

We present the semantics of Open Atomic Ethernet (\oae{}) links as a concrete
realization of a non-\fito{} protocol architecture.  The key insight is that
causal order is not assumed \emph{a priori} but \emph{created} through
transaction structure: the link state machine progresses through
\textsc{tentative}~$\to$~\textsc{reflecting}~$\to$~\textsc{committed},
with the option to abort at any point before commitment.  Delivery does not
imply commitment; commitment requires reflective acknowledgment---proof that
information has round-tripped and been semantically validated by both endpoints.

We formalize this through three frameworks.  First, the \oae{} link state
machine, a six-state finite automaton whose normative invariants guarantee
that semantic corruption cannot occur at the link level.  Second, Indefinite
Logical Timestamps (\ilt{}), a four-valued causal structure that admits a
genuinely indefinite relation ($a \leftrightarrow b$) between concurrent
events, resolving only after symmetric link-level exchange.  Third, the
Slowdown Theorem applied to links, which establishes that round-trip
measurement is the \emph{minimum} interaction required to establish causal
order---the semantic arrow emerges at the transaction boundary, not below it.

We show that \ilt{} is strictly more expressive than Definite Causal Order
(\dco{}) systems for reversible link protocols, and that standard verification
tools (TLA$^{+}$, P~language) cannot faithfully represent the indefinite
phase without collapsing it to nondeterminism.  We connect these results to
the Knowledge Balance Principle from quantum information theory, showing that
the \oae{} link's epistemic registers implement Spekkens' toy model at the
protocol level.

The paper concludes with a comparative analysis showing that \oae{} achieves
infinite consensus number---enabling wait-free implementations of arbitrary
concurrent objects---while RDMA, NVLink, and UALink remain limited to finite
consensus numbers due to their \fito{} semantics.
\end{abstract}

\section{Introduction}
\label{sec:intro}

Part~I of this series argued that computing harbors a hidden arrow of
time---not the thermodynamic arrow of entropy increase, but a
\emph{semantic} arrow embedded in the \fito{} assumption that pervades
every layer of the networking stack.  We showed that this assumption
constitutes a category mistake in the sense of Ryle~\cite{ryle1949}:
treating an epistemic convention (the logical ordering of messages) as
an ontic commitment (a law of physical causality).

This paper asks: \emph{what does a protocol architecture look like when
the category mistake is corrected?}

The answer is not merely theoretical.  Open Atomic Ethernet (\oae{})
is a protocol framework that embodies three principles derived from
the physics of time:
\begin{enumerate}[leftmargin=1.2cm]
  \item \textbf{Interaction-derived causality:} causal order is not
    assumed but \emph{created} through the structure of transactions.
  \item \textbf{Reflective acknowledgment:} the minimum interaction
    required to establish causal order is a round-trip---a message
    and its semantic reflection.
  \item \textbf{Reversibility until commitment:} state transitions
    are tentative until both parties confirm; irreversibility enters
    only at the point of commitment.
\end{enumerate}

These principles are not engineering preferences.  They follow from the
physics reviewed in Part~I: the time-symmetry of fundamental laws
(Roberts~\cite{roberts2022}), the reality of indefinite causal order
(Oreshkov et al.~\cite{oreshkov2012}), and Spekkens' demonstration
that epistemic restrictions can reproduce quantum phenomena without
ontic commitments~\cite{spekkens2007}.  If the correct ontology is
one where causal order emerges from interaction rather than being
imposed by background structure, then protocols should be designed
accordingly.

The remainder of this paper is organized as follows.
Section~\ref{sec:shannon} traces the \fito{} assumption to its origin in
Shannon's channel model and shows how it propagated through Lamport's
happened-before relation to become an implicit axiom of distributed
computing.
Section~\ref{sec:atomicity} presents the atomicity framework---three
orthogonal dimensions of failure that the semantic arrow must address.
Section~\ref{sec:state-machine} specifies the \oae{} link state machine
and its normative invariants.
Section~\ref{sec:ilt} introduces Indefinite Logical Timestamps and the
four-valued causal structure.
Section~\ref{sec:slowdown} applies the Slowdown Theorem to establish
the round-trip as the minimum causal interaction.
Section~\ref{sec:kbp} connects the link's register structure to
Spekkens' Knowledge Balance Principle.
Section~\ref{sec:comparison} provides a comparative analysis with
RDMA, NVLink, UALink, and CXL.
Section~\ref{sec:summary} summarizes and previews Part~III.

\section{Shannon's Channel Model and the Origin of \fito{}}
\label{sec:shannon}

Claude Shannon's 1948 paper~\cite{shannon1948} models communication
as a unidirectional process: a source selects a message, a transmitter
encodes it, a channel corrupts it with noise, and a receiver decodes
it.  The channel capacity theorem establishes the maximum rate at which
information can be reliably transmitted in the forward direction.

Shannon was solving a specific problem---maximizing throughput on
telephone lines---and his unidirectional model was appropriate for
that purpose.  But when the model was adopted as the foundation for
networking protocols, a subtle promotion occurred.  The unidirectional
channel became not merely an engineering abstraction but the
\emph{ontological model} of communication: information flows from
sender to receiver, and the return path is auxiliary---an
acknowledgment mechanism, not a constitutive element of the
communication itself.%
\footnote{Every real Ethernet link is physically bidirectional:
full-duplex operation has been standard since Fast Ethernet (1995).
Yet protocol design treats the return path as overhead---the ACK as
a reliability mechanism, not as a semantic constituent of the
interaction.  The arrow of time on the link is imposed by the
protocol designer, not by the physics of the medium.}

Lamport's happened-before relation~\cite{lamport1978} extended
Shannon's \fito{} to distributed systems theory.  By defining
$a \rightarrow b$ whenever $a$ is the sending of a message and $b$ is
its receipt, Lamport encoded the assumption that information flow is
unidirectional: the sender's state can influence the receiver's, but
not conversely.  As shown in the companion
paper~\cite{borrill2026-lamport}, this encodes three implicit
commitments: temporal monotonicity, asymmetric causation, and global
causal reference.

Edward Lee has articulated the deeper issue:
\begin{quote}
\small
Determinism is a property of models, not of the physical
world.~\cite{lee2009}
\end{quote}
Shannon's channel is a \emph{model} of communication, not communication
itself.  The model is deterministic and unidirectional; the physical
medium is neither.  Promoting the model's properties to ontological
status---treating the channel's forward direction as the direction
of causality---is the category mistake that \oae{} corrects.

\section{The Atomicity Framework}
\label{sec:atomicity}

Before specifying the link state machine, we must identify what it
protects against.  The semantic arrow of time is violated when
transactions fail to preserve meaning across their execution.
Following the analysis in~\cite{borrill2026-semantic-integrity}, we
identify three orthogonal dimensions of failure:

\begin{definition}[Three Failure Classes]
\label{def:failure-classes}
\mbox{}
\begin{enumerate}[leftmargin=1.2cm]
  \item \textbf{Detected Unrecoverable Errors (DUE):} failures that
    are detected and cause an explicit fail-stop.  These preserve the
    semantic arrow: the system knows it has failed and can act
    accordingly.
  \item \textbf{Silent Data Corruption (SDC):} undetected value
    corruption.  The semantic arrow is violated silently: data parses
    correctly but means the wrong thing.
  \item \textbf{Atomicity Errors:} violations of atomic visibility,
    agreement, or completeness.  These are the \emph{dominant} failure
    mode in distributed systems and the primary target of the \oae{}
    link state machine.
\end{enumerate}
\end{definition}

Atomicity errors are further decomposed into three dimensions, each
corresponding to a way the semantic arrow can be violated:

\begin{definition}[Three Atomicity Dimensions]
\label{def:atomicity-dims}
\mbox{}
\begin{enumerate}[leftmargin=1.2cm]
  \item \textbf{Atomicity of Updates:} state transitions must occur
    as indivisible operations.  No observer may see a partially
    applied update.
  \item \textbf{Atomicity of Communication:} delivery of data does
    not constitute semantic agreement.  A message may arrive
    intact---every bit correct---and yet the receiver may not have
    committed the semantic transition that the sender intended.
  \item \textbf{Atomicity of Visibility:} the state observed by any
    participant must correspond to a consistent snapshot.  Mixed-era
    state---some fields from before a transaction, some from after---must
    not be observable.
\end{enumerate}
\end{definition}

\footnote{The OCP Silent Data Corruption
initiative~\cite{ocp-sdc2024} has focused primarily on SDC---the
hardware-level corruption of bit values.  The \oae{} framework extends
this to atomicity errors, which are semantic rather than syntactic.
A system can have perfect bit integrity and still suffer catastrophic
semantic corruption if atomicity is violated.}

These dimensions yield three normative invariants, stated in RFC~2119
language~\cite{rfc2119}:

\begin{description}[leftmargin=1.2cm]
  \item[A1.] Systems \textsc{must not} expose partially applied updates.
    Global updates \textsc{must} commit atomically or not at all.
  \item[A2.] Delivery confirmation \textsc{must not} be treated as
    semantic agreement.  Receivers \textsc{must} validate completeness
    and semantic consistency before committing.
  \item[A3.] Observed state \textsc{must} correspond to a consistent
    snapshot.  Mixed-era state \textsc{must not} be treated as final.
\end{description}

These invariants define what the semantic arrow \emph{requires}.  The
link state machine, specified next, defines how to \emph{enforce} them.

\section{The \oae{} Link State Machine}
\label{sec:state-machine}

The \oae{} link state machine is a six-state finite automaton that
governs the lifecycle of every transaction on a link.  Its design
embodies the three principles stated in the introduction:
interaction-derived causality, reflective acknowledgment, and
reversibility until commitment.

\begin{definition}[\oae{} Link States]
\label{def:link-states}
The link state machine has six states:
\begin{description}[leftmargin=2cm]
  \item[\textsc{reset}] Link initialization or recovery.  No semantic
    assumptions are valid.  No application-visible state exchange
    occurs.
  \item[\textsc{idle}] Quiescent, synchronized link.  Ready for a new
    interaction.  No tentative state is outstanding.
  \item[\textsc{tentative}] Provisional transmission of semantic
    content.  Data has been sent but \textsc{must not} be treated as
    committed by either endpoint.  The receiver treats incoming data
    as provisional.
  \item[\textsc{reflecting}] Explicit semantic reflection.  The
    receiver reflects back an interpreted digest of the content,
    enabling the sender to validate that semantic agreement has been
    reached.  Reflection \textsc{must not} itself commit state.
  \item[\textsc{committed}] Atomic acceptance.  The transition to
    \textsc{committed} must be explicit.  All tentative state becomes
    visible atomically.  Once committed, \textsc{must not} be silently
    revoked.
  \item[\textsc{aborted}] Explicit rejection or rollback.  All
    tentative state is discarded.  No partial effects remain visible.
\end{description}
\end{definition}

The mandatory transitions are:
\begin{align}
  \textsc{reset}      &\longrightarrow \textsc{idle} \label{eq:reset-idle}\\
  \textsc{idle}       &\longrightarrow \textsc{tentative} \label{eq:idle-tent}\\
  \textsc{tentative}  &\longrightarrow \textsc{reflecting} \label{eq:tent-refl}\\
  \textsc{reflecting} &\longrightarrow \textsc{committed} \label{eq:refl-comm}\\
  \textsc{tentative}  &\longrightarrow \textsc{aborted} \label{eq:tent-abort}\\
  \textsc{reflecting} &\longrightarrow \textsc{aborted} \label{eq:refl-abort}\\
  \textsc{committed}  &\longrightarrow \textsc{idle} \label{eq:comm-idle}\\
  \textsc{aborted}    &\longrightarrow \textsc{idle} \label{eq:abort-idle}
\end{align}

\footnote{Note what is \emph{absent}: there is no direct transition
from \textsc{tentative} to \textsc{committed}.  Every commitment must
pass through \textsc{reflecting}---the round-trip is mandatory, not
optional.  This is the structural enforcement of reflective
acknowledgment: you cannot commit what you have not confirmed.}

The link state machine enforces four normative invariants:

\begin{enumerate}[leftmargin=1.2cm]
  \item The link \textsc{must not} expose committed state unless
    \emph{both} endpoints have reached \textsc{committed}.
  \item The link \textsc{must not} infer commitment from delivery,
    completion, or ordering guarantees.
  \item Loss, reordering, or duplication of messages \textsc{must}
    result in reflection failure or explicit abort.
  \item Timeouts \textsc{must} resolve to \textsc{aborted}, not
    to implicit commit.
\end{enumerate}

Invariant~(4) is the direct negation of \fito{}'s timeout-and-retry
(TAR) pattern.  In a \fito{} protocol, a timeout triggers a retry in
the forward direction; the system assumes the original message was lost
and tries again.  In \oae{}, a timeout triggers an abort: the system
assumes that the transaction has failed and returns to a known-good
state.  Forward progress is achieved not by retrying the same operation
but by initiating a \emph{new} transaction from quiescence.

\section{Indefinite Logical Timestamps}
\label{sec:ilt}

Classical distributed systems offer two causal relations between
events: \emph{ordered} ($a \rightarrow b$ or $b \rightarrow a$) and
\emph{concurrent} ($a \parallel b$).  The concurrent relation is a
residual category: it means ``we cannot determine the order,'' but the
framework assumes that a definite order \emph{exists}---it is merely
unknown.%
\footnote{Lamport~\cite{lamport1978}: ``Two distinct events $a$
and $b$ are said to be concurrent if $a \nrightarrow b$ and
$b \nrightarrow a$.''  Concurrency is defined negatively, as the
absence of ordering.  It does not admit the possibility that the
ordering is genuinely indefinite.}

Indefinite Logical Timestamps (\ilt{})~\cite{borrill2026-ilt}
introduce a fourth causal relation that corresponds to the indefinite
causal order established by Oreshkov, Costa, and
Brukner~\cite{oreshkov2012} in quantum mechanics:

\begin{definition}[Four-Valued Causal Structure]
\label{def:four-valued}
For events $a$ and $b$ on an \oae{} link, the causal relation is one of:
\begin{align}
  a &\prec b    &&\text{(definite forward: $a$ causally precedes $b$)} \\
  b &\prec a    &&\text{(definite backward: $b$ causally precedes $a$)} \\
  a &\parallel b &&\text{(concurrent: causally independent)} \\
  a &\leftrightarrow b &&\text{(indefinite: order not yet resolved)}
\end{align}
The indefinite relation $a \leftrightarrow b$ is temporary and
reversible.  It resolves into one of the other three relations only
after a symmetric link-level exchange---the reflecting phase of the
link state machine.
\end{definition}

The indefinite relation is not epistemic uncertainty (we don't know
the order) but ontic indefiniteness (the order does not yet exist).
This mirrors the distinction in quantum mechanics between a mixed
state (classical ignorance) and a superposition (genuine
indefiniteness).  The resolution of $a \leftrightarrow b$ into a
definite relation is analogous to the collapse of a superposition
upon measurement---except that in \oae{}, the ``measurement'' is the
reflecting phase of the transaction.%
\footnote{The 2-bit encoding from the Leibniz Bridge
synthesis~\cite{borrill2026-leibniz}: \texttt{01}~$= a \prec b$,
\texttt{10}~$= b \prec a$, \texttt{00}~$= a \parallel b$,
\texttt{11}~$= a \leftrightarrow b$.  The gap in classical
distributed systems is cell~\texttt{11}: no existing framework
(Lamport, Fidge, Mattern) can represent it.}

\subsection{Tensor Clocks}

\ilt{} replaces vector clocks with a more compact representation:

\begin{definition}[Tensor Clock]
\label{def:tensor-clock}
A tensor clock on a link between nodes $A$ and $B$ is a triple
$\mathit{TC} = (c_A, c_B, d)$, where:
\begin{itemize}[leftmargin=1.2cm]
  \item $c_A$ is $A$'s induced forward influence (how many
    transactions $A$ has initiated),
  \item $c_B$ is $B$'s induced forward influence (how many
    transactions $B$ has initiated),
  \item $d$ is the number of symmetric rounds that have reached
    mutual knowledge (reflecting phases completed).
\end{itemize}
\end{definition}

Tensor clocks avoid the $O(n)$ space blowup of vector clocks by
confining state to the Moore neighborhood in a mesh topology.  Each
link maintains its own tensor clock independently; there is no global
clock vector to maintain or synchronize.  This locality is essential
for scalability: the state required per link is constant regardless
of system size.

\subsection{Reversible Protocol Algebra}

The lifecycle of a transaction on an \ilt{}-equipped link consists of
three phases:
\begin{enumerate}[leftmargin=1.2cm]
  \item \textbf{Indefinite phase:} hyperdata circulates on the link;
    orientation is not fixed; the causal relation between the endpoints'
    operations is $a \leftrightarrow b$.
  \item \textbf{Resolution phase:} one side commits to initiating;
    the orientation flag is set; the reflecting phase begins.
  \item \textbf{Completion:} indefiniteness collapses into a definite
    causal order ($a \prec b$ or $b \prec a$).  Classical \dco{}
    machinery (logs, audit trails, TLA$^{+}$ models) can safely project
    the result.
\end{enumerate}

\begin{theorem}[Expressiveness]
\label{thm:expressiveness}
\ilt{} is strictly more expressive than \dco{} systems in the domain
of reversible link protocols.
\end{theorem}

\begin{proof}[Proof sketch]
We exhibit a protocol state (the indefinite phase, where
$a \leftrightarrow b$) that can be represented in \ilt{} but not in
any \dco{} framework.  In \dco{}, every pair of events is either
ordered or concurrent; there is no third category.  Since the
indefinite relation is neither ordered nor concurrent (it resolves into
either upon completion), it cannot be faithfully represented.  Formal
details in~\cite{borrill2026-ilt}.
\end{proof}

\begin{lemma}
\label{lem:tla}
No TLA$^{+}$ specification can faithfully represent the \ilt{}
indefinite phase without collapsing it to nondeterminism.
\end{lemma}

\begin{lemma}
\label{lem:p-language}
No P~language machine can encode reversible \ico{}-style causal
symmetry as native semantic structure.
\end{lemma}

\footnote{The inability of TLA$^{+}$ and P to represent \ilt{} is
not a criticism of these tools---they are excellent for what they
do.  It is a demonstration that \fito{} is embedded in the
\emph{metalogic} of our verification tools, not just in the protocols
they verify.  A new generation of verification tools, capable of
expressing indefinite causal order natively, is needed.}

These results have practical consequences.  AWS's systems correctness
practices~\cite{aws-correctness2025}, which include TLA$^{+}$, P
models, and theorem provers, are all fully \dco{}/\fito{} in their
semantic core.  They can verify that a protocol behaves correctly
\emph{given} that causality is forward-only and globally acyclic.
They cannot verify that a protocol correctly handles the indefinite
phase---the very phase that makes \oae{} links more expressive than
\fito{} protocols.

\section{The Slowdown Theorem Applied to Links}
\label{sec:slowdown}

The Slowdown Theorem of Gorard, Beggs, and
Tucker~\cite{gorard-beggs-tucker} establishes a lower bound on
computational irreducibility in physical systems.  Applied to a
classical scatter machine---a projectile scattered by a sharp wedge
with binary outcomes---the theorem shows that excavating arbitrarily
many decimal places of initial-condition precision requires
correspondingly many computational steps.  No shortcut exists.

The application to links is direct.  Consider two endpoints, Alice
and Bob, connected by a physical medium.  Alice sends a message; Bob
receives it.  In a \fito{} framework, this is sufficient to establish
causal order: $\text{send}(m) \rightarrow \text{recv}(m)$.  But what
has actually been established?  Alice knows she sent a message.  She
does not know whether Bob received it, whether Bob interpreted it as
she intended, or whether Bob's state is now consistent with hers.%
\footnote{This is the Two Generals problem in its purest form.
Gray~\cite{gray1978} showed that no finite number of messages
suffices to guarantee coordinated action.  But the problem is
deeper than Gray realized: it is not that coordination is
\emph{impossible} but that forward-only messages are the wrong
primitive for achieving it.}

The Slowdown Theorem tells us why: \emph{establishing the causal
relationship between two events requires at least one round-trip.}
A single message establishes \emph{placement}---data has been
transmitted.  Only a reflected message---one that demonstrates the
receiver has processed the content and fed back a semantic
digest---establishes \emph{agreement}.

This is the network-level analogue of the scatter machine's
requirement for multiple interactions to excavate precision.  In the
scatter machine, a single observation yields one bit (left or right);
extracting $n$ bits requires $n$ interactions.  On a link, a single
message yields placement; extracting semantic agreement requires a
round-trip.

\begin{proposition}[Round-Trip Minimum]
\label{prop:roundtrip}
The minimum interaction required to establish a definite causal
ordering between two events on a link is one round-trip (message
plus reflected response).  This is a consequence of the Slowdown
Theorem: the causal precision of a single message is insufficient
to determine mutual semantic state.
\end{proposition}

\subsection{Subtime and the Invisible Structure}

The Slowdown Theorem reveals a \emph{subtime} structure within links
that is invisible to coarse-grained observation.  In the scatter
machine, a projectile can be trapped at the wedge point in a perpetual
ping-pong---behavior that occurs ``below'' the resolution of the
measurement apparatus.  On a link, the analogous phenomenon is the
circulation of hyperdata---pre-frame bits that travel the link
continuously but never cross the PCIe boundary~\cite{borrill2026-ilt}.

This subtime structure is not an engineering curiosity; it is the
substrate on which the indefinite causal relation is maintained.  The
indefinite phase of a transaction corresponds to the period during
which hyperdata circulates without resolving into a definite orientation.
The resolution of indefiniteness---the transition from
$a \leftrightarrow b$ to $a \prec b$---corresponds to the projectile
escaping the wedge: a measurement has been performed, and the
previously subtime dynamics have been projected into observable time.%
\footnote{The correspondence is precise: in the scatter machine,
the trapped projectile carries \emph{information} (it has been
scattered) but this information is \emph{invisible} until the
projectile escapes.  On the link, hyperdata carries semantic content
that is invisible to the application until the reflecting phase
resolves.  Both are resources---usable, conserved, but hidden.}

\subsection{Perfect Information Feedback}

When the frame length on a link exceeds the link distance, a
remarkable phenomenon occurs: acknowledgments arrive while transmission
is still in progress.  Munamala et al.~\cite{munamala2025} call this
\emph{Perfect Information Feedback} (\pif{}): the link becomes its own
temporal reference, with the round-trip time shorter than the
transmission time.

\pif{} collapses the distinction between ``sending'' and ``receiving''
into a continuous bidirectional exchange.  There is no meaningful sense
in which one endpoint is the ``sender'' and the other the ``receiver'';
both are simultaneously transmitting and reflecting.  This is the
physical realization of the principle that causal order emerges from
interaction rather than being imposed by protocol convention.

Under \pif{} conditions, the \oae{} link state machine can complete
the \textsc{tentative}~$\to$~\textsc{reflecting}~$\to$~\textsc{committed}
cycle within a single frame time.  The semantic arrow is established
without any temporal gap between proposal and confirmation---the
transaction is ``born committed'' because the round-trip completes
before the forward transmission does.

\section{The Knowledge Balance Principle at the Link Level}
\label{sec:kbp}

Rob Spekkens' toy model~\cite{spekkens2007} demonstrated that many
features of quantum mechanics---including entanglement, no-cloning,
teleportation, and the uncertainty principle---can be reproduced by a
classical model supplemented with a single epistemic restriction: the
\emph{Knowledge Balance Principle} (\kbp{}).  The principle states that
in any state of maximal knowledge about a system, the amount known
must equal the amount unknown.%
\footnote{For a system with $2N$ bits of ontic state, an epistemic
agent can know at most $N$ bits.  This single constraint, applied
consistently, generates an extraordinary range of ``quantum''
phenomena in a purely classical framework.}

The Leibniz Bridge synthesis~\cite{borrill2026-leibniz} applies \kbp{}
directly to \oae{} link design.  Each link endpoint maintains two
classes of registers:

\begin{description}[leftmargin=2cm]
  \item[\textsc{ont} registers:] represent the \emph{imagined} ontic
    state of the link: four bits (two bits per direction, two
    directions).  These registers are a theoretical construct---no
    single endpoint can ever access the full ontic state.  The ONT
    register is the state the link ``would have'' if an omniscient
    observer could see both sides simultaneously; no such observer
    exists in any physical implementation.
  \item[\textsc{epi} registers:] hold the \emph{accessible} epistemic
    state---what the endpoint \emph{knows} about the link.  Two bits:
    exactly half the ontic state, per \kbp{}.  This is all any
    endpoint can ever know.
\end{description}

The \kbp{} constraint means that each endpoint always knows exactly
half of the link's state.  The other half is unknown---not because
of noise or latency, but because the epistemic restriction is
\emph{constitutive} of the protocol.  Knowing more than half would
violate \kbp{} and break the analogy with quantum mechanics; knowing
less than half would leave the endpoint unable to participate in
transactions.

This has a concrete consequence for the link state machine.  In
the \textsc{tentative} state, the sender knows it has proposed a
transition (one bit of epistemic state) but does not know whether
the receiver has accepted (the other bit is unknown).  In the
\textsc{reflecting} state, the receiver's reflection provides the
missing bit, satisfying \kbp{} at both endpoints simultaneously.
The transition to \textsc{committed} occurs precisely when both
endpoints' epistemic states are maximally informed---each knows
half the ontic state, and the two halves are consistent.

\begin{proposition}[Knowledge Balance on Links]
\label{prop:kbp}
The \oae{} link state machine implements a classical instantiation
of Spekkens' \kbp{}: at every stage of a transaction, each
endpoint's epistemic state contains exactly half the information
needed to determine the transaction's outcome.  Full information
(sufficient for commitment) requires contributions from both
endpoints.
\end{proposition}

This is why unilateral commitment---the \fito{} pattern of
``send and assume received''---is not merely unreliable but
\emph{epistemically incoherent}: it claims to know more about
the link's state than \kbp{} permits.

\section{Comparative Analysis}
\label{sec:comparison}

We now compare the \oae{} link semantics with four existing
interconnect technologies: RDMA (InfiniBand and RoCE), NVLink,
UALink, and CXL.  The comparison is organized around the semantic
arrow: which technologies preserve meaning across transactions,
and which permit semantic corruption.

\subsection{RDMA: Completion Without Commitment}

RDMA's one-sided operations (Write, Read, Atomic) are the purest
expression of \fito{} at the hardware level.  A Remote Write
completes when the NIC signals that data has been placed in the
receiver's memory.  But ``placed'' does not mean ``committed'':
the receiving application may not have seen the data; the data may
be in a NIC buffer, not yet visible through the cache hierarchy;
and multi-field updates are not atomic beyond 8 bytes.%
\footnote{Ziegler et al.~\cite{ziegler2024} provide the first
comprehensive analysis of synchronization bugs in one-sided RDMA,
finding that concurrent cache-line retrievals can result in logically
invalid combinations of fields---despite successful completion
signals.  Completion succeeded; meaning was destroyed.}

In the \oae{} framework, RDMA's completion corresponds to
the \textsc{tentative} state: data has been proposed but not
confirmed.  RDMA has no \textsc{reflecting} state and no
\textsc{committed} state.  The NIC's completion signal is
treated as commitment---a violation of invariant~A2 (delivery
must not be treated as semantic agreement).

\subsection{NVLink: Signal-Based Ordering}

NVLink's semantic model is stronger than RDMA's.  Its
signal-based ordering guarantees that when a signal operation
arrives at the destination, all preceding put operations on the
same context are visible.  This is a form of release-acquire
semantics: the signal serves as a one-directional commitment
point.

In the \oae{} framework, NVLink's signal corresponds to a
partial \textsc{reflecting} phase: the sender knows the data has
been placed and made visible.  But the signal is unilateral---the
receiver does not reflect back a semantic digest.  If the receiver
interprets the data differently from the sender's intent, the
divergence is undetectable.  NVLink provides atomicity of
visibility (A3) but not atomicity of communication (A2).

\subsection{CXL: Cache Coherence Without Transaction Semantics}

CXL~3.0's back-invalidation mechanism~\cite{cxl30} provides
hardware-enforced cache coherence at the cache-line level.  When a
CXL device modifies a cache line, it issues a back-invalidate snoop
to all host caches that may hold stale copies.  This ensures that
every reader sees the latest version of each individual cache line.

But cache-line coherence is not transaction coherence.  A
multi-cache-line update is not atomic: readers may observe some
lines from before the update and some from after.  CXL addresses
atomicity of visibility (A3) at the cache-line granularity but not
at the transaction granularity.  Like RDMA, CXL has no
\textsc{reflecting} phase and no mechanism for semantic agreement
between endpoints.

\subsection{Consensus Numbers}

The theoretical foundation for comparing these technologies is
Herlihy's consensus hierarchy~\cite{herlihy1991}.  An object's
\emph{consensus number} is the maximum number of processes for which
the object can solve wait-free consensus.

\begin{theorem}[Consensus Number of \oae{}]
\label{thm:consensus}
\oae{} primitives have infinite consensus number, enabling wait-free
implementations of arbitrary concurrent objects.
\end{theorem}

\begin{proof}[Proof sketch]
\oae{} can implement $n$-process consensus via compare-and-swap on
a decision register ordered by the Global Ordering Service.  Since
\oae{} transactions are linearizable and the reversibility engine
guarantees progress (failed transactions abort to quiescence rather
than blocking), the construction satisfies Herlihy's wait-freedom
requirement for all $n$.  Full proof in~\cite{borrill2026-oae}.
\end{proof}

\begin{table}[h]
\small
\begin{center}
\begin{tabular}{lcccc}
\toprule
& \textbf{\oae{}} & \textbf{RDMA} & \textbf{NVLink} & \textbf{CXL} \\
\midrule
Consensus number & $\infty$ & 2 & 2 & 2 \\
Atomicity of updates (A1) & \checkmark & partial & \checkmark & partial \\
Atomicity of communication (A2) & \checkmark & --- & --- & --- \\
Atomicity of visibility (A3) & \checkmark & --- & \checkmark & line-level \\
Reflecting phase & mandatory & absent & absent & absent \\
Reversibility & native & --- & --- & --- \\
\bottomrule
\end{tabular}
\end{center}
\caption{Semantic comparison of interconnect technologies.  Only \oae{}
enforces all three atomicity invariants and provides a reflecting
phase.}
\label{tab:comparison}
\end{table}

\section{Summary and Preview of Part~III}
\label{sec:summary}

This paper has presented the semantic foundations of Open Atomic
Ethernet links.  The central claim is that the semantic arrow of
time---the direction in which meaning is preserved across
transactions---is not imposed by protocol convention but
\emph{created} by transaction structure.

The \oae{} link state machine enforces this through three mechanisms:
the mandatory reflecting phase (no commitment without round-trip
confirmation), the four-valued causal structure of \ilt{} (admitting
genuinely indefinite causal order during the tentative phase), and
the Knowledge Balance Principle (each endpoint knows exactly half the
link state, ensuring that commitment requires bilateral contribution).

The Slowdown Theorem establishes that these mechanisms are not
engineering preferences but physical necessities: round-trip
measurement is the minimum interaction required to establish causal
order, and the subtime structure of links provides the substrate on
which indefinite causality is maintained.

Part~III~\cite{borrill2026-partIII} examines what happens when these
principles are violated at industrial scale: RDMA's ``completion
fallacy'' and its consequences for AI training at 32,000 GPUs.


\section*{Acknowledgments}
This paper was developed with AI assistance (Claude, Anthropic) for literature organization, drafting, and \LaTeX{} preparation. The research program originates from the author's work on hardware-verified deterministic networking (VERITAS, 2005), replicated state machines (REPLICUS, 2008), atomic Ethernet (Earth Computing, 2012--2020), and the FITO category mistake framework (D{\AE}D{\AE}LUS, 2024--present). All technical claims, interpretations, and arguments are the sole responsibility of the author.


\end{document}